\documentclass{article}
\usepackage{amssymb}
\usepackage{amsmath}
\usepackage{hyperref}	

\usepackage{graphicx}
\usepackage[caption=false,subrefformat=parens,labelformat=parens]{subfig}
\usepackage{amsthm}
\theoremstyle{plain}
\newtheorem{theorem}{Theorem}
\newtheorem{lemma}{Lemma}
\renewenvironment{proof}{{\noindent\bfseries Proof.}}{$\blacksquare$}
\begin{document}

\title{Autophosphorylation and the dynamics of the activation of Lck}

\author{Lisa Maria Kreusser\\
Department for Applied Mathematics and Theoretical Physics\\  
University of Cambridge\\
Wilberforce Road\\
Cambridge CB3 0WA, UK\\
and\\
Alan D. Rendall\\
Institut f\"ur Mathematik\\
Johannes Gutenberg-Universit\"at\\
Staudingerweg 9\\
D-55099 Mainz\\
Germany}

\date{}

\maketitle

\begin{abstract}
Lck (lymphocyte-specific protein tyrosine kinase) is an enzyme which plays
a number of important roles in the function of immune cells. It belongs to the  
Src family of kinases which are known to undergo autophosphorylation. It turns
out that this leads to a remarkable variety of dynamical behaviour which can
occur during their activation. We prove that in the presence of
autophosphorylation one phenomenon, bistability, already occurs in a
mathematical model for a protein with a single phosphorylation site. We
further show that a certain model of Lck exhibits oscillations.
Finally we discuss the relations of these results to models in the literature
which involve Lck and describe specific biological processes, such as the early
stages of T cell activation and the stimulation of T cell responses resulting
from the suppression of PD-1 signalling which is important in immune
checkpoint therapy for cancer.
\end{abstract}  

\section{Introduction}

Phosphorylation and dephosphorylation, the processes in which proteins are
modified by the addition or removal of phosphate groups, play an important
role in biology. The activity of an enzyme is influenced by its
phosphorylation state and 
these processes provide a way of switching enzymes on and off quickly. The 
enzymes which catalyse phosphorylation and dephosphorylation are called kinases
and phosphatases, respectively. The phosphorylation of a protein X is usually
catalysed by another protein Y. It may also be catalysed by X itself, a
process called autophosphorylation. This can happen either in {\it trans} (one
molecule of X catalyses the phosphorylation of a site on another molecule of
X) or in {\it cis} (a molecule of X catalyses the phosphorylation of a site on
that same molecule). Here we are concerned with the kinase Lck
\cite{bommhardt19}, which can undergo both autophosphorylation in {\it trans}
and phosphorylation by another kinase Csk. Lck belongs to the Src family
of kinases \cite{shah18} which have many properties in common, in particular
those related to their phosphorylation.

In what follows we are interested in understanding the way in which the
activity of Lck is controlled, an issue which is important for analysing how
the function of immune cells is regulated. More specifically, we want to do so
by studying mathematical models for phosphorylation processes. There has been
a lot of work on models for cases where there is a clear distinction between
substrates and enzymes. A standard example is the multiple futile cycle
where bounds for the maximal number of steady states were obtained in
\cite{wang08} and \cite{flockerzi14} and for the maximal number of stable
steady states in \cite{feliu19}. Much less is known in the case of
autophosphorylation.
To our knowledge the earliest papers on mathematical modelling of Src family
kinases are by Fu\ss\ et al. \cite{fuss06}, \cite{fuss08}. In the first of
these papers the authors consider a system coupling Src (with
autophosphorylation included) to Csk and the phosphatase PTP$\alpha$. They
then introduce a simplification by assuming the concentration of Csk to be
constant, and find two fold bifurcations in simulations. In particular, this
system appears to exhibit bistability. In \cite{fuss08} sustained oscillations
and infinite period bifurcations were observed in a slight extension of the
model of \cite{fuss06}. These dynamical features occurred in a context where the
basic system describing phosphorylation and dephosphorylation of Src is
embedded in feedback loops. In fact it was found in \cite{kaimachnikov09}
that complicated dynamical behaviour is possible even without the feedback
loops. More recently the dynamics of a model for autophosphorylation of a
protein with only one phosphorylation site was studied in \cite{doherty15}. In
that case also two fold bifurcations were observed. The model considered there
is one-dimensional and thus relatively easy to analyse. The bistability
found in \cite{doherty15} contrasts with the situation in the multiple futile
cycle where in the case of a single phosphorylation site there is only one
steady state.

In Section \ref{sec:2} a model for autophosphorylation is introduced which is of central
importance in what follows and it is shown that in a certain Michaelis-Menten
limit it can be reduced to a one-dimensional model. Section \ref{sec:3} contains an
analysis of some properties of solutions of this reduced model. In particular
it is shown that this system can exhibit more than one stable steady state.
This section provides a rigorous treatment of some features found in the
simulations of \cite{doherty15}. The property of bistability is lifted to the
original model. The main results are Theorems \ref{th:1}-\ref{th:3}. The model of Section \ref{sec:2}
without external kinase only exhibits bistability under the condition that 
phosphorylation has an activating effect on the enzyme. The corresponding 
case with inhibition exhibits no multistability. The aim of Section \ref{sec:4} is to 
show that in the case of an inhibitory phosphorylation multistability can be 
restored by modelling the external kinase explicitly. The main result is Theorem \ref{th:4}. 
Here, in contrast to the 
results of Section \ref{sec:3}, the multistability is not present in the 
Michaelis-Menten limit.

Section \ref{sec:5} is concerned with a model for Lck which can be reduced by timescale
separation to a two-dimensional one. The original model inherits certain
patterns of behaviour such as bistability, Hopf bifurcations and homoclinic
orbits from the two-dimensional one. It is proved that the two-dimensional
model does exhibit these phenomena as a consequence of the occurrence of a
Bogdanov-Takens bifurcation. The main result is Theorem \ref{th:5}. In Section \ref{sec:6} the
models analysed in the present paper are compared with ones which occur as
parts of more comprehensive models in the literature describing some concrete
biological situations. Section \ref{sec:7} presents some ideas on possible further
developments of the results of this paper.

\section{The basic model}\label{sec:2}

Consider a protein with one phosphorylation site. We denote the
unphosphorylated form of this protein by X and the phosphorylated form by
Y. Suppose X is able to catalyse its own phosphorylation in {\it trans}.
The simplest model for this reaction is 2X$\to$ X+Y. If Y is also able to
catalyse the phosphorylation of X then this can be modelled by the reaction
X+Y$\to$ 2Y. The basic model considered in what follows includes these two
reactions together with phosphorylation of X catalysed by a kinase E and
dephosphorylation of Y catalysed by a phosphatase F. Mass action kinetics is
assumed for the autophosphorylation reactions. For the other two processes we
use a description with mass action kinetics involving a substrate, an enzyme
and a complex, which we call an extended Michaelis-Menten description.

The concentrations of X, Y, E, F and the complexes XE and YF are denoted by
$x$, $y$, $e$, $f$, $d$ and $c$, respectively. The evolution equations are of
the form
\begin{eqnarray}
  &&\dot x=-k_1x^2+k_4c-k_5ex+k_6d-k_8xy,\label{basic1}\\
  &&\dot d=k_5ex-(k_6+k_7)d,\label{basic2}\\
  &&\dot e=-k_5ex+(k_6+k_7)d,\label{basic3}\\
  &&\dot c=k_2fy-(k_3+k_4)c,\label{basic4}\\
  &&\dot f=-k_2fy+(k_3+k_4)c,\label{basic5}\\
  &&\dot y=k_1x^2-k_2fy+k_3c+k_7d+k_8xy,\label{basic6}
\end{eqnarray}
where the dot stands for the derivative with respect to $t$ and the $k_i$ are
positive reaction constants. There are three
conserved quantities defined by the total amounts of the substrate and the two
enzymes $E$ and $F$. These are $A=x+c+d+y$, $B=c+f$ and $C=d+e$. A situation
where the amounts of both enzymes and the rates of both autophosphorylation 
reactions are small can be described using a
Michaelis-Menten reduction. To do this, introduce new variables by means of
the relations $k_1=\epsilon\tilde k_1$, $k_8=\epsilon\tilde k_8$,
$c=\epsilon\tilde c$, $f=\epsilon\tilde f$, $d=\epsilon\tilde d$,
$e=\epsilon\tilde e$ and $\tau=\epsilon t$. Substituting these relations into
the above equations and dropping the tildes gives
\begin{eqnarray}
  &&x'=-k_1x^2+k_4c-k_5ex+k_6d-k_8xy,\label{basicr1}\\
  &&y'=k_1x^2-k_2fy+k_3c+k_7d+k_8xy,\label{basicr2}\\
  &&\epsilon d'=k_5ex-(k_6+k_7)d,\label{basicr3}\\
  &&\epsilon e'=-k_5ex+(k_6+k_7)d,\label{basicr4}\\
  &&\epsilon c'=k_2fy-(k_3+k_4)c,\label{basicr5}\\
  &&\epsilon f'=-k_2fy+(k_3+k_4)c,\label{basicr6}
\end{eqnarray}
where the prime stands for the derivative with respect to $\tau$. If we set
$\epsilon=0$ in these equations the last four become algebraic.
Combining these with the conservation laws and doing the usual algebra for
Michaelis-Menten reduction leads to the relations $c=\frac{By}{K_{M1}+y}$ and
$d=\frac{Cx}{K_{M2}+x}$ where $K_{M1}=\frac{k_3+k_4}{k_2}$ and
$K_{M2}=\frac{k_6+k_7}{k_5}$. It follows that 
\begin{equation}\label{scalar}
x'=-k_1x^2+\frac{Bk_4y}{K_{M1}+y}-\frac{Ck_7x}{K_{M2}+x}-k_8{xy}
\end{equation}
while $y$ satisfies an analogous equation. These two equations are equivalent
because $x+y$ is a conserved quantity for $\epsilon=0$. Thus the whole dynamics
is contained in the single equation (\ref{scalar}) in that case. When $C=0$ (no
external kinase) the equation for $y$ reduces (up to a difference of notation)
to the equation (1) in \cite{doherty15}. To make it clear that this is an
equation for a single unknown it is necessary to use the conserved quantity
$A=x+y$. Thus for $C=0$ the evolution equation for $y$ is
\begin{equation}\label{scalar2}
y'=k_1(A-y)^2+k_8(A-y)y-\frac{Bk_4y}{K_{M1}+y}.
\end{equation}

\section{Analysis of the model of Doherty et al.}\label{sec:3}

In \cite{doherty15} the authors describe certain aspects of the dynamics of
solutions of equation (\ref{scalar2}). Here we complement their analysis by
giving rigorous proofs of some of these. Steady states of this equation are
zeroes of the polynomial
\begin{align}
p_3(y)&=[(k_1-k_8)y^2+(-2k_1A+k_8A)y+k_1A^2](K_{M1}+y)-Bk_4y\nonumber\\
&= -(\alpha-1)y^3+[-K_{M1}(\alpha-1)+A(\alpha-2)]y^2\nonumber\\
&\quad+[K_{M1}A(\alpha-2)+A^2-Bk_1^{-1}k_4]y+ K_{M1}A^2    
\end{align}
where $\alpha=\frac{k_8}{k_1}$. Positive steady states of the evolution
equations for $x$ and $y$ are in one to one correspondence with roots of this
polynomial in the interval $(0,A)$. Note that $p_3(0)>0$ and $p_3(A)<0$. If
$k_1-k_8>0$ then $p_3$ must have one root with $x<0$ and one with $x>A$. Thus
it has exactly one root in the biologically relevant region. When $k_1-k_8<0$
there could be up to three roots in $(0,A)$. Since no root can cross the
endpoints of the interval the number of roots counting multiplicity is odd
for any values of the parameters. In biological terms, bistability
is only possible when phosphorylation activates the enzyme. In the case of
Lck there are two phosphorylation sites of central importance for the
regulation of the kinase activity, Y394 and Y505, whose phosphorylation
is activatory and inhibitory, respectively. Thus if we wanted to use 
this model to describe Lck with mutations targeting one of its 
phosphorylation sites then to have a chance of bistability it is the 
inhibitory site Y505 which should be knocked out. This type of modification
of Lck has been studied experimentally in \cite{amrein88}. It was discovered
that the mutated protein exhibits carcinogenic effects.

It will now be shown that there is a region in parameter space where three
positive steady states exist.

\begin{theorem}\label{th:1}
If $\alpha>2$, $A^2k_1<Bk_4$, $k_8$ is sufficiently large and
$K_{M1}$ is sufficiently small for fixed values of the other parameters then
the equation (\ref{scalar2}) has three hyperbolic steady states, of which two
are asymptotically stable and the other unstable.
\end{theorem}

\begin{proof}
When three steady states exist they must be simple zeros of $p_3$
and it follows that when ordered by the value of $y$ the first and third
steady states are stable while the second is unstable. Each of these steady
states is hyperbolic. Thus to complete the proof of the theorem it suffices to
prove the existence of three steady states under the given assumptions.
The condition for a steady state can be written in the form $q_1(y)=q_2(y)$
where $q_1(y)=(A-y)[k_1A+(-k_1+k_8)y]$ and $q_2(y)=\frac{Bk_4y}{K_{M1}+y}$. Note
that $q_2(y)<Bk_4$ for all $y\ge 0$. If $\alpha>1$ then $q_1$ has a local
maximum when $y=y_1=\frac{(\alpha-2)A}{2(\alpha-1)}$. Assume that $\alpha>2$ so
that $y_1>0$. Evaluating at the maximum gives
$q_1(y_1)=\frac{k_8\alpha A^2}{4(\alpha-1)}$. By choosing
$k_8$ large enough while keeping all other parameters fixed we can ensure that
this maximum is greater that $Bk_4$. It follows that $q_1(y_1)>q_2(y_1)$. Then
choosing $K_{M1}$ small enough while keeping all other parameters fixed
and using the fact that $A^2k_1<Bk_4$ ensures that there is some $y_2<y_1$ with
$q_1(y_2)<q_2(y_2)$. This implies that there are two roots of $p_3$ which
are less than $y_1$ and these are simple. Under these conditions $p_3$ has
three positive roots in the interval $(0,A)$ and so there exist three positive
steady states. 
\end{proof}

\noindent
This theorem and its proof are illustrated by Fig.\ \ref{fig:theorem1} where we show $q_1$, $q_2$ and $p_3$ for parameters $A = 1$, $B = 2$, $k_1 = 1$, $k_4 = 1$, $k_8 = 8$, $\alpha = \frac{k_8}{k_1}$ and $K_{M1} = 0.02$ satisfying the assumptions in Theorem \ref{th:1}.
\begin{figure}[htbp]
	\centering
	\includegraphics[width=0.7\textwidth]{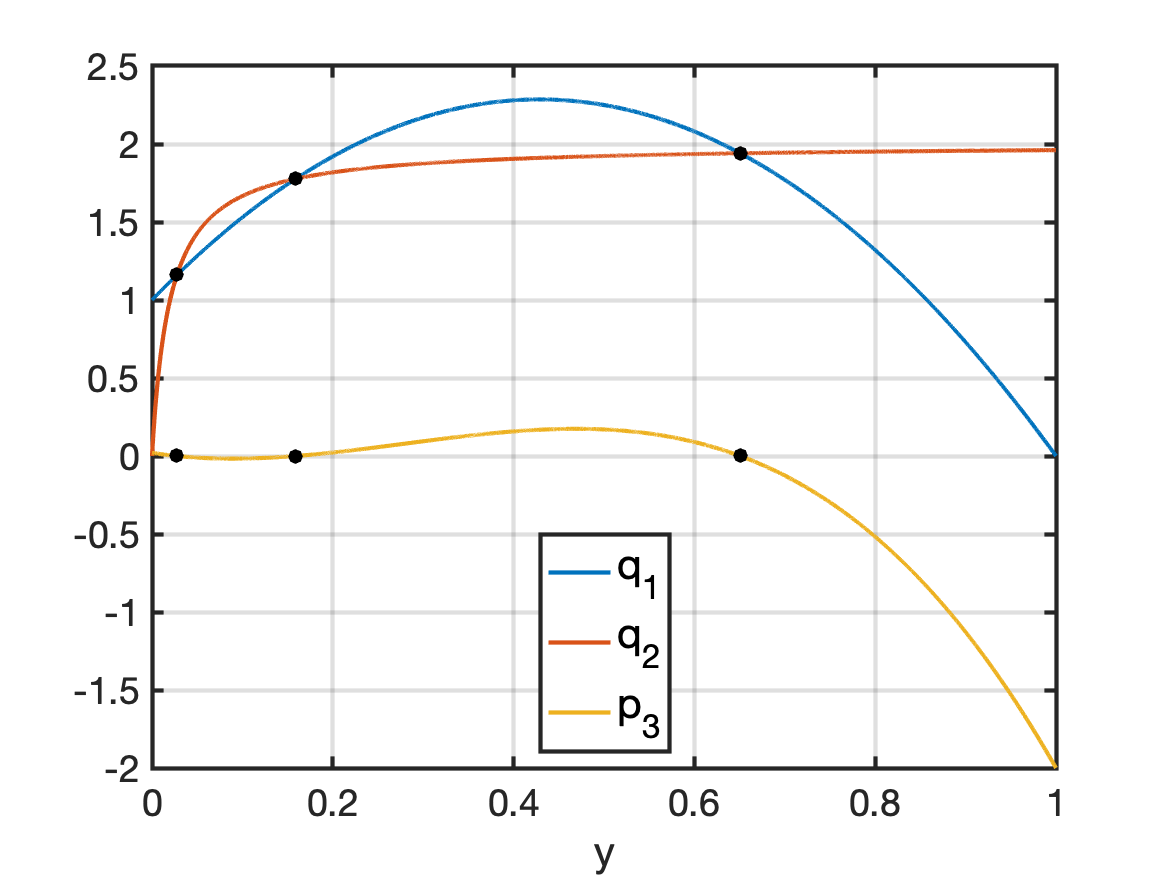}
	\caption{Illustration of Theorem \ref{th:1}\label{fig:theorem1}.}
\end{figure}

In fact the three steady states arise in a single bifurcation. To prove this
we first need a result on cubic equations.

\begin{lemma}\label{lem:1}
	The polynomial $p(x)=ax^3+bx^2+cx+d$ has a triple root if and only
	if $b^3=27a^2d$ and $c^3=27ad^2$.
\end{lemma}

\begin{proof}
If $x_*$ is a triple root then $p(x_*)=p'(x_*)=p''(x_*)=0$. From
the last of these equations we can conclude that $x_*=-\frac{b}{3a}$.
Substituting this in the other two equations gives $b^2=3ac$ and
$b^3=\frac92a(bc-3ad)$. Combining the last two equations gives $b^3=27a^2d$
and $c^3=27ad^2$. Suppose conversely that $b^3=27a^2d$ and $c^3=27ad^2$. Then
$bc=9ad$ and $abc=9a^2d$. Thus $\frac92(abc-3a^2d)=27a^2d$ and it follows that
$b^3=\frac92a(bc-3ad)$. Using $b^3=27a^2d$ then implies that $b^2=3ac$. With
all this information it can be checked directly that $x_*=-\frac{b}{3a}$ is
a triple root of $p$. 
\end{proof}

\begin{theorem}\label{th:2}
The three steady states in Theorem \ref{th:1} arise in a generic cusp
bifurcation.
\end{theorem}

\begin{proof}
To prove this it will be shown that the parameters can be chosen
so that the polynomial $p_3$ satisfies the conditions of Lemma \ref{lem:1}. Assume that
$K_{M1}<A$. Since $b^3=[-K_{M1}(\alpha-1)+A(\alpha-2)]^3$ we see that
$b^3=(A-K_{M1})^3\alpha^3+\ldots$ for $\alpha$ large and
$b^3=(K_{M1}-2A)^3+\ldots$ for $\alpha\to 0$. Now $A-K_{M1}>0$ and  $K_{M1}-2A<0$.
Thus if we consider $b^3$ as a function of $\alpha$ with the other parameters
fixed it is an increasing function which takes on all values in the interval
$[(K_{M1}-2A)^3,\infty)$. On the other hand $a^2d=(\alpha-1)^2K_{M1}A^2$ and so
$a^2d=K_{M1}A^2\alpha^2+\ldots$ for $\alpha$ large and $a^2d=K_{M1}A^2+\ldots$ for
$\alpha\to 0$. It follows that there exists an $\alpha_*$ for which 
$b^3=27a^2d$. In this way the first condition of Lemma \ref{lem:1} has been achieved.
Since $a^2d$ is non-negative there the same must be true of $b$ and it follows
that $\alpha_*>2$. Hence $ad^2$ is negative and so in order to achieve the
second condition of the Lemma \ref{lem:1} it is enough to show that $c$ can
be given any prescribed negative value by choosing $k_4$ appropriately while
fixing the other parameters. Note that $a$, $b$ and $d$ do not depend on $k_4$
so that the first condition remains satisfied. Since $\alpha_*>2$ the
quantity $c$ is positive for $\alpha=\alpha_*$ and $k_4$ sufficiently small.
By increasing $k_4$ it can then be made to have any desired negative value.
Thus it can be ensured that the second condition is satisfied. Note that
the point $x_*$ at which the bifurcation takes place does lie in the
biologically relevant region $(0,A)$ since there is one steady state in that
region and $x_*$ is, neglecting multiplicity, the only one.

Next we note that the derivative of the mapping
$(\alpha,K_{M1},A,k_4)\mapsto (a,b,c,d)$ is always invertible for $\alpha>2$.
Thus by the inverse function theorem we see that by varying the parameters
arbitrarily we can vary the coefficients of the polynomial $p_3$ arbitrarily in
a neighbourhood of the values for the triple root. Thus we can choose two
parameters so that the point with the triple root is embedded in a generic
cusp bifurcation as defined in \cite{kuznetsov10}. More specifically we can 
choose a mapping $(\beta_1,\beta_2)\mapsto (\alpha,K_{M1},A,k_4)$ such that, 
after translating the coordinate $y$ so that the bifurcation is at the origin,
we have $(a,b,c,d)=(1,0,\beta_2,\beta_1)$. 
\end{proof}

Consider now the rescaled mass action system (\ref{basicr1})-(\ref{basicr6})
in the case $C=0$. In this case we can discard the equations for $d$ and $e$.
Moreover, we can use the conservation laws to discard the equations for $x$
and $c$ and replace these quantities in the right hand sides of the equations
for $y$ and $f$. The result is
\begin{eqnarray}
&&y'=k_1(A-B-y+f)^2-k_2fy+k_3(B-f)+k_8xy,\label{2dsystem1}\\
&&\epsilon f'=-k_2fy+(k_3+k_4)(B-f)\label{2dsystem2}.
\end{eqnarray}
We now want study the limit $\epsilon\to 0$ in these equations and show that
solutions converge.

\begin{theorem}\label{th:3}
There is a choice of parameters such
that the system (\ref{basic1})-(\ref{basic6}) with $d=e=0$, $C=0$ and 
fixed values of $A$ and $B$ imposed has three steady states, of which two are
asymptotically stable and the other a hyperbolic saddle. The
three steady states arise in a generic cusp bifurcation. For arbitrary values
of the parameters each solution converges to a steady state as $t\to\infty$.
In particular, this system has no periodic solutions.
\end{theorem}

\begin{proof}
It suffices to prove corresponding results for the system
(\ref{2dsystem1})-(\ref{2dsystem2}). The theorem can be proved using the
results of Theorems \ref{th:1} and \ref{th:2} and geometric singular perturbation theory
(GSPT) \cite{kuehn15}. The important condition to be checked is that of normal
hyperbolicity. It says that on the critical manifold, which is the zero set of
the right hand side in the equation for $f'$, the derivative of that right
hand side with respect to $f$ should be non-zero. This is indeed the case
since the derivative is $-k_2f-k_3-k_4<0$. It can be concluded that for each
hyperbolic steady state of the Michaelis-Menten system there is a nearby
steady state of the mass action system which is hyperbolic within the
invariant manifold of constant $A$ and $B$. In addition, when the steady state
of the Michaelis-Menten system is stable the same is true of the corresponding
steady state of the mass action system and when the steady state of the
Michaelis-Menten system is unstable the steady state of the mass action system
is a saddle point whose stable manifold is one-dimensional. To obtain the
statement about the convergence of general solutions to steady states we
compute the linearization of (\ref{2dsystem1})-(\ref{2dsystem2}) which is
\begin{equation}
A=\left[
{\begin{array}{cc}
2k_1(f+y-A-B)-k_2f & 2k_1(f+y-A-B)-k_2y-k_3\\
-\epsilon^{-1}k_2 f& -\epsilon^{-1}[k_2y+(k_2+k_3)]
\end{array}}
\right]
\end{equation}
It is always the case that $A-y$ and $B-f$ are positive on the region of
biological interest. Thus the system is competitive. Every solution of a
competitive two dimensional system converges to a steady state \cite{smith95}
and this completes the proof of the theorem. 
\end{proof}

To conclude this section we consider the limiting case of the system
(\ref{scalar2}) obtained by setting $k_1=0$. In this case only the
phosphorylated form of the protein is catalytically active. Bistability for a
system of this type was considered in \cite{lisman85}. If we continue to
assume $C=0$ then $y=0$ is a steady state. Thus in order to get
bistability we need to include that boundary steady state in the counting.
With this understanding we obtain an analogue of Theorem \ref{th:1} for this case,
where the condition on $\alpha$ is absent. The proof is strictly analogous to
that of Theorem \ref{th:1}. To see what happens to Theorem \ref{th:2} in this case we need to
replace $p_3$, which was got by division by $k_1$, by
$\tilde p_3=y[k_8(-y+A)(y+K_{M1})-Bk_4]$. This polynomial has a triple root at
the origin when $A=K_{M1}$ and $AK_{M1}=Bk_4$.

\section{Effect of an external kinase}\label{sec:4}

We next consider the case where the phosphorylated kinase is completely
inactive, which can be modelled by setting $k_8=0$ in the model of the last 
section. This might be thought of as a model of the mutant of Lck where the
activatory site Y394 is knocked out. It should, however, be noted that
in reality the catalytic activity of this mutant, although much reduced, is
not actually zero \cite{smith93}. In that case we have $k_1-k_8>0$ and, as
mentioned above, there is only one positive steady state in the
Michaelis-Menten system. Next we will investigate the case where $k_1-k_8>0$ but
an external kinase is present ($C>0$). It turns out that there is still only
one steady state in the Michaelis-Menten system. For in any such steady state
we have
\begin{equation}
k_1x^2+k_8x(A-x)+\frac{Ck_7x}{K_{M2}+x}=\frac{Bk_4(A-x)}{K_{M1}+A-x}.
\end{equation}
Since the function on the left hand side of this equation is monotone
increasing on $[0,A]$ and is zero for $x=0$ while the function on the right
hand side is monotone decreasing on $[0,A]$ and is zero for $x=A$ these two
functions are equal at a unique point $x\in (0,A)$. Thus there cannot be more
than one steady state in the Michaelis-Menten system with $k_8<k_1$. It turns
out, however, that there can be more than one steady state in the corresponding
mass action system, even in the case $k_8=0$.

Positive solutions of the mass action system with $k_8=0$ are in one to one
correspondence with solutions of the following system obtained by using the
conserved quantities to eliminate $d$, $c$ and $y$.
\begin{eqnarray}
&&\dot x=-k_1x^2+k_4(B-f)-k_5ex+k_6(C-e),\label{basicelim1}\\
&&\dot e=-k_5ex+(k_6+k_7)(C-e),\label{basicelim2}\\
&&\dot f=-k_2f(A-B-C-x+e+f)+(k_3+k_4)(B-f).\label{basicelim3}
\end{eqnarray}

Define a polynomial by $p_6(x)=\sum_{i=0}^6a_ix^i$ with coefficients 
\begin{eqnarray}
a_6&=&k_1^2k_2k_5^2,\\
a_5&=&2k_1^2k_2k_5(k_6+k_7)+k_1k_2k_4k_5^2,\\
a_4&=&k_1k_2[-k_5^2((A+B)k_4-(k_4+2k_7)C-k_4(k_6+k_7))\nonumber\\
&&+k_1(k_6+k_7)^2+k_1k_5(k_6+k_7)]-k_1k_4(k_3+k_4)k_5^2,\\
a_3&=&k_1k_5(k_6+k_7)\{k_2[-2(A+B)k_4+(k_4+2k_7)C]\nonumber\\
&&-2(k_3+k_4)k_4\}+k_2k_4[k_5^2(Bk_4-k_7C)+k_1(k_6+k_7)^2],\\
a_2&=&k_2[-k_1k_4(A+B)(k_6+k_7)^2-k_4^2k_5(k_6+k_7)B\nonumber\\
&&+k_5^2(k_4A-(k_4+k_7)C-k_4(k_6+k_7))(Bk_4-k_7C)]\nonumber\\
&&-(k_3+k_4)k_4[k_1(k_6+k_7)^2+k_5^2k_7C],\\
  a_1&=&k_2k_4k_5(k_6+k_7)[B(Ak_4-(k_4+k_7)C-k_4(k_6+k_7))\nonumber\\
&&+A(k_4B-k_7C)]-(k_3+k_4)k_4k_5(k_6+k_7)k_7C,\\
a_0&=&k_2k_4^2(k_6+k_7)^2AB.      
\end{eqnarray}
Define $x_{\rm max}$ to be the largest value of $x$ satisfying the inequalities
\begin{eqnarray}
&&\frac{k_1}{k_4}x^2+\frac{k_5k_7Cx}{k_4(k_5x+k_6+k_7)}\le B,\label{xmax1}\\
&&x+\frac{k_1}{k_4}x^2+\left(1+\frac{k_7}{k_4}\right)
\frac{k_5Cx}{(k_5x+k_6+k_7)}\le A.\label{xmax2}
\end{eqnarray}
Note that $x_{\rm max}$ depends continuously on the parameters.

\begin{lemma}\label{lem:2}
For given positive values of $A$, $B$ and $C$ positive steady
state solutions of the system (\ref{basic1})-(\ref{basic6}) with $k_8=0$ are
in one to one correspondence with roots of the polynomial $p_6$ in the
interval $(0,x_{\rm max})$.
\end{lemma}

\begin{proof}
Note first that the equations for steady states of
(\ref{basic1})-(\ref{basic6}) are equivalent to the equations for steady
states of (\ref{basicelim1})-(\ref{basicelim3}) and that these in turn
are equivalent to the equations
\begin{eqnarray}
&&k_5ex=(k_6+k_7)(C-e),\label{sfca1}\\
&&k_2f(A-B-C-x+e+f)=(k_3+k_4)(B-f),\label{sfca2}\\
&&k_1x^2=k_4(B-f)-k_7(C-e).\label{sfca3}
\end{eqnarray}

Now suppose that $(x,d,e,c,f,y)$ is a positive steady state. It follows from
(\ref{sfca1}) that $e=\frac{(k_6+k_7)C}{k_5x+k_6+k_7}$ and combining this with
(\ref{sfca3}) gives $f=B-(k_1/k_4)x^2-\frac{k_5k_7Cx}{k_4(k_5x+k_6+k_7)}$. Thus
we have solved for $e$ and $f$ in terms of $x$. Substituting this information
into (\ref{sfca2}) and rearranging gives the equation $p_6(x)=0$. Suppose
conversely that $x$ is a root of $p_6$ with $0<x<x_{\rm max}$. Define
\begin{eqnarray}
&&e=\frac{(k_6+k_7)C}{k_5x+k_6+k_7},\label{backsolve1}\\
&&f=B-(k_1/k_4)x^2-(k_7/k_4)(C-e),\label{backsolve2}\\
&&y=A-B-C-x+e+f.\label{backsolve3}
\end{eqnarray}
It follows from (\ref{xmax1}) that the quantity $f$ defined by
(\ref{backsolve2}) is positive and from (\ref{xmax2}) that the quantity $y$
defined by (\ref{backsolve3}) is positive. It follows directly that
(\ref{sfca1}) and (\ref{sfca3}) hold. The fact that $x$ is a root of $p_6$
implies that (\ref{sfca2}) holds and hence that $(x,d,e,c,f,y)$ is a positive
solution of (\ref{basic1})-(\ref{basic6}). 
\end{proof}

Consider now the real roots of $p_6$. They depend continuously on the 
parameters and their number is constant modulo two. Since $p_6(0)>0$ a root of
$p_6$ cannot pass through zero. We claim that $x_{\rm max}$
can also never be a root of $p_6$. For if $x$ were equal to $x_{\rm max}$
while satisfying the inequalities (\ref{xmax2}) then at
least one of them would become an equality. In the first case $f$ as defined
by (\ref{backsolve2}) would be equal to zero. This contradicts equation
(\ref{sfca2}). In the second case $y$ defined by equation (\ref{backsolve3})
would be equal to zero. But then it follows from (\ref{basic4}) that $c=0$ and
from (\ref{basic6}) that $x=0$, a contradiction. It can be concluded that the
sign of $p_6(x_{\rm max})$ is independent of
the parameters. To determine what the sign is it suffices to evaluate it
for some particular values of the parameters. Choose $k_i=1$ for all $i$,
$A=5$, $B=2$ and $C=3$. When $x=1$ we see that equality holds in
(\ref{xmax1}) while the strict inequality holds in (\ref{xmax2}). Thus in this
case $x_{\rm max}=1$. Evaluating the coefficients in $p_6$ gives $a_6=1$, $a_5=5$,
$a_4=8$, $a_3=-15$, $a_2=-43$, $a_1=-18$, $a_0=40$. Hence $p_6(x_{\rm max})=-22<0$.
It follows from the intermediate value theorem that $p_6$ has a least one root
in each of the intervals $(0,x_{\rm max})$ and $(x_{\rm max},\infty)$.
The number of sign changes of the coefficients in the polynomial is
even and at most four. Thus Descartes' rule of signs implies that the number
of positive roots is zero, two or four. The case with no positive roots has
already been ruled out. Thus there are two or four and at least one of them
must be greater that $x_{\rm max}$. With the parameter values in the example
there are only two changes of sign, only two positive roots and we know that
precisely one is less than $x_{\rm max}$. By continuity the number of roots in
$(0,x_{\rm max})$ counting multiplicity is odd. It can only be one or three and
we have already seen an example of parameters where it is one. In that case
the system (\ref{basic1})-(\ref{basic6}) admits precisely one positive steady
state.

We will show that there also exist parameter values such that the
system has three positive steady states. One approach would be to show that
there are parameters for which there is a triple root in the desired interval
and then perturb. In fact we will show directly that there are parameters for
which there are three roots in that interval, since that approach is simpler.

\begin{theorem}\label{th:4}
There is a choice of parameters for which the system
(\ref{basic1})-(\ref{basic6}) with $k_8=0$ has three positive steady states.
\end{theorem}

\begin{proof}
Due to Lemma \ref{lem:2} it suffices to find parameter values for which the
polynomial has three roots in the interval $(0,x_{\rm max})$. It follows from
the preceding discussion that it is enough to show that the interval contains
at least two roots. It turns out that is suffices to choose $A=6$, $B=20$,
$C=2$, $k_i=1$ for all $i\ne 5$ and $k_5$ sufficiently large. With these
choices it follows that we get the following asymptotics for
$k_5\to\infty$. $a_6=k_5^2+\ldots$, $a_5=k_5^2+\ldots$, $a_4=-20k_5^2+\ldots$
$a_3=18k_5^2+\ldots$, $a_2=-4k_5^2+\ldots$ where the terms not written explicitly
are $o(k_5^2)$, as are the coefficients $a_1$ and $a_0$. It follows that
$k_5^{-2}x^{-2}p(x)=q(x)+o(1)$, where $q(x)=x^4+x^3-20x^2+18x-4$. In the limit
the inequalities defining the admissible interval become
$x^2<18$ and $x+x^2<2$. By continuity it suffices to show that $q$ has two
roots in the interval $(0,1)$. This is true because $q(0)<0$, $q(1/2)>0$ and
$q(1)<0$. 
\end{proof}

\section{Analysis of a model for wild-type Lck}\label{sec:5}

The results of the previous sections were related to situations in which
one of the two key regulatory phosphorylation sites in a Src family kinase
such as Lck is mutated. In the present section we move to the case where
both sites are present. The starting point for the discussion is the model
introduced in \cite{kaimachnikov09}. There four phosphorylation states of
the kinase are included in the description. The first, denoted by $S_i$, is
that where the inhibitory site is phosphorylated while the activatory site
is not. This form of the kinase shows no catalytic activity. $S$, $S_{a1}$ and
$S_{a2}$ are the forms where neither site is phosphorylated, only the
activatory site is phosphorylated and both sites are phosphorylated,
respectively. All of these are catalytically active to some extent and can
catalyse the transition $S\to S_{a1}$. The transitions $S\to S_i$ and
$S_{a1}\to S_{a2}$ are catalysed by Csk. The transitions $S_i\to S$ and
$S_{a2}\to S_{a1}$ are catalysed by one phosphatase and the transitions
$S_{a1}\to S$ and $S_{a2}\to S_i$ are catalysed by another phosphatase.
Experimental results obtained in \cite{hui14} indicate that some modifications
of these assumptions may be needed to obtain a biologically correct model. In
particular, it was found that Y505 in Lck undergoes autophosphorylation in
{\it trans}, albeit with a much lower rate than Y394. The variants of the
model of \cite{kaimachnikov09} which would be needed to take this into account
will not be considered further in the present paper - the aim here is rather
to see the variety of dynamical behaviour which this type of system can produce.

Let us introduce the following neutral notation for the quantities involved
in the model, denoting the concentrations of $S$, $S_i$, $S_{a1}$ and
$S_{a2}$ by $x_1$, $x_2$, $x_3$ and $x_4$ respectively. Then $X=x_1+x_2+x_3+x_4$
is a conserved quantity. The reaction rate for the autophosphorylation is
bilinear, the dephosphorylation of $S_{1a}$ is given by Michaelis-Menten
kinetics and the other reactions are assumed to be linear. Using the notations
of \cite{kaimachnikov09} for the reaction constants gives the system
\begin{eqnarray}
&&\dot x_1=-k_2x_1+k_1x_2+k_4\frac{x_3}{\beta+x_3}-k_3x_1(\delta x_1+x_3+x_4),
\label{kaim1}\\
&&\dot x_2=k_2x_1-k_1x_2+k_7x_4,\label{kaim2}\\
&&\dot x_3=k_3x_1(\delta x_1+x_3+x_4)-k_4\frac{x_3}{\beta+x_3}
+k_6x_4-k_5x_3,\label{kaim3}\\  
&&\dot x_4=k_5x_3-(k_6+k_7)x_4.\label{kaim4}
\end{eqnarray}

Before considering this system in the general case note that
setting $x_2$, $x_4$, $k_2$ and $k_5$ to zero reduces this system to
\begin{eqnarray}
&&\dot x_1=k_4\frac{x_3}{\beta+x_3}-k_3x_1(\delta x_1+x_3),\\
&&\dot x_3=k_3x_1(\delta x_1+x_3)-k_4\frac{x_3}{\beta+x_3}.
\end{eqnarray}
Either of the variables can be eliminated using the conserved quantity giving
an equation which is, up to a difference in notation, exactly the equation of
Doherty et al. discussed in previous sections. Only setting $k_2$ and $k_5$ to
zero in (\ref{kaim1})-(\ref{kaim4}) gives a partially decoupled system which
is the product of the system of \cite{doherty15} with a hyperbolic saddle. It
follows immediately from Theorem \ref{th:1} that for suitable values of the parameters
the system (\ref{kaim1})-(\ref{kaim4}) admits at least three positive steady
states, of which two are stable and hyperbolic and the third is a hyperbolic
saddle.

We now return to the general system (\ref{kaim1})-(\ref{kaim4}).
In \cite{kaimachnikov09} the authors find a remarkable variety of dynamic
behaviour in the system above which, after fixing a value of the conserved
quantity, is of dimension three. They remark that there is a limiting case
which gives rise to a system of dimension two which already exhibits a lot
of this dynamics. To investigate this possibility we define a new variable by
$y=x_3+x_4$ and use it to replace $x_3$. In addition we introduce rescaled
parameters satisfying $\tilde k_5=\epsilon k_5$ and $\tilde k_6=\epsilon k_6$.
Making these substitutions and discarding the tildes leads to the system
\begin{align}
\dot x_1&=-k_2x_1+k_1x_2+k_4\frac{y-x_4}{\beta+y-x_4}
-k_3x_1(\delta x_1+y),
\label{kaiml1eps}\\
\dot x_2&=k_2x_1-k_1x_2+k_7x_4,\label{kaiml2eps}\\
\dot y&=k_3x_1(\delta x_1+y)
     -k_4\frac{y-x_4}{\beta+y-x_4}-k_7x_4,
     \label{kaiml3eps}\\
\epsilon\dot x_4&=k_5(y-x_4)-(k_6+\epsilon k_7)x_4.\label{kaiml4eps}
\end{align}
This is a fast-slow system with one fast and three slow variables.
We have the conserved quantity $X=x_1+x_2+y$. In the limiting case
$\epsilon=0$ equation (\ref{kaiml4eps}) reduces to $y-x_4=\xi x_4$, where
$\xi=\frac{k_6}{k_5}$. It follows that $x_4=\frac{1}{1+\xi}y$.
Substituting this into (\ref{kaiml1eps}) and (\ref{kaiml2eps}) and using the
conserved quantity gives the following system of two equations:
\begin{align}
\dot x_1&=-k_2x_1+k_1x_2+k_4\frac{\xi (X-x_1-x_2)}{\beta(\xi+1)+\xi(X-x_1-x_2)}
     \nonumber\\
&\qquad-k_3x_1[(X-x_1-x_2)+\delta x_1],\label{kaimlr1}\\
\dot x_2&=k_2x_1-k_1x_2+k_7\frac{1}{\xi+1}(X-x_1-x_2)\label{kaimlr2}.    
\end{align}

In the terminology of GSPT this is the
restriction of the system to the critical manifold. This critical manifold
is normally hyperbolic and stable since the partial derivative of the right
hand side of (\ref{kaiml4eps}) with respect to $x_4$ evaluated at $\epsilon=0$
is negative. This allows us to transport information about stability and
bifurcations from steady states of the two-dimensional system to steady
states of the full system. Positive steady states of the full system
with a given value of $X$ are in one to one correspondence with positive
steady states of (\ref{kaimlr1})-(\ref{kaimlr2}) with $x_1+x_2<X$. At
steady state equation (\ref{kaimlr2}) can be used to express $x_2$ in
terms of $x_1$ and substituting this into (\ref{kaimlr1}) shows that
$x_1$ is a root of a cubic polynomial which is not identically zero. Thus
the system (\ref{kaimlr1})-(\ref{kaimlr2}) has at most three steady states.

It will be shown that the system (\ref{kaimlr1})-(\ref{kaimlr2}) admits
periodic solutions which arise in a Hopf bifurcation and homoclinic orbits.
In order to do this it suffices to show that this system admits a generic
Bogdanov-Takens bifurcation \cite{kuznetsov10}. By saying that the bifurcation
is generic we mean that it satisfies the conditions BT.0, BT.1, BT.2 and BT.3 of
\cite{kuznetsov10}. Then the desired results follow from Theorem 8.5 of
\cite{kuznetsov10} and the analysis of the normal form of the bifurcation
preceding that theorem. Let $J(x_1,x_2)$ be the linearization of the system
(\ref{kaimlr1})-(\ref{kaimlr2}) about the point $(x_1,x_2)$. Finding a
bifurcation point where the condition BT.0 is satisfied means finding a point
$(x_1,x_2)$ and a choice of the parameters of the system so that $J(x_1,x_2)$
has a double zero eigenvalue but is not itself zero. If the right hand sides
of equations (\ref{kaimlr1}) and (\ref{kaimlr2}) are denoted by $f_1$ and $f_2$
this means solving the system of four equations given by the vanishing of
$f_1$, $f_2$, ${\rm tr} J$ and $\det J$. The general strategy is to choose all
but four of the parameters and use the four equations to solve for the rest.
An obstacle to this is that the quantities resulting from this process
might fail to be positive. This obstacle was overcome by trial and error.

The equations for steady states can be written in the following form.
\begin{align}
k_4&=\frac{[k_2x_1-k_1x_2+k_3x_1(X-x_1-x_2+\delta x_1)]\label{k4sol}
[\beta(\xi+1)+\xi(X-x_1-x_2)]}{\xi(X-x_1-x_2)},\\
k_7&=\frac{(\xi+1)(-k_2x_1+k_1x_2)}{X-x_1-x_2}.\label{k7sol}
\end{align}
The linearization is
\begin{equation}
J=\left[
{\begin{array}{cc}
   -k_2-\phi(\beta)-k_3(X-2x_1-x_2+2\delta x_1)
   & k_1-\phi(\beta)+k_3x_1\\
\eta& -\omega
\end{array}}
\right]
\end{equation}
Here we have introduced the auxiliary quantities $\eta=k_2-\frac{k_7}{\xi+1}$,
$\omega=k_1+\frac{k_7}{\xi+1}$ and
$\phi(\beta)=\frac{k_4\xi (\xi+1)\beta}{[\beta (\xi+1)+\xi (X-x_1-x_2)]^2}$.
Suppose that we have a Bogdanov-Takens bifurcation. Since the trace is zero we
have that the first element in the first row of the Jacobian must be equal to
$\omega$. Hence
\begin{equation}
\phi(\beta)+\omega+k_2=-k_3[X-2(1-\delta)x_1-x_2].\label{betadef}
\end{equation}
It follows that
\begin{equation}
k_1-\phi(\beta)+k_3x_1
=\omega+k_1+k_2+k_3[X-(1-2\delta)x_1-x_2].
\end{equation}  
Since the determinant is zero we have
\begin{equation}\label{quadratic}
\omega^2+\eta[\omega+k_1+k_2+k_3(X-x_1-x_2+2\delta x_1)]=0.
\end{equation}

Choose $X=\frac32$, $x_1=1$, $x_2=\frac14$, $k_1=8$, $k_2=1$, $\delta=\frac16$,
$\xi=1$. It follows from (\ref{k7sol}) that $k_7=8$. Putting this into
(\ref{quadratic}) gives $k_3=\frac{324}{7}$. It then follows from
(\ref{betadef}) that $\phi(\beta)=\frac{44}{7}$. On the other hand
(\ref{k4sol}) implies that $k_4=\frac{128}{7}(8\beta+1)$. Combining this with
the definition of $\phi$ shows that $\beta=\frac{11}{936}$. Finally we compute
$k_4=\frac{16384}{819}$. The conclusion is that with the given choices there
is exactly one solution for the remaining parameters $(k_7,k_3,\beta,k_4)$
such that the system satisfies the condition BT.0 for a Bogdanov-Takens
bifurcation at the chosen point with coordinates
$(x_1,x_2)=\left(1,\frac14\right)$. At this point the linearization is of
the form
\begin{equation}
J=\left[
{\begin{array}{cc}
   12& 48\\
-3& -12
\end{array}}
\right].
\end{equation}
When talking about a Bogdanov-Takens bifurcation we need a system depending
on two parameters. In our example we choose these to be $\delta$ and $k_3$
and consider all other parameters in the system as fixed

As will now be explained, a calculation
shows that conditions BT.1, BT.2 and BT.3 are also satisfied so that this is
a generic Bogdanov-Takens bifurcation. For this purpose it is convenient to
transform to coordinates $y_1=-\frac13 x_2+\frac1{12}$ and $y_2=x_1+4x_2-2$
adapted to the eigenvectors of $J$. Then $\dot y_i=(J_0y)_i+Q_i(y)+O(|y|^3)$,
where $J_0$ is in Jordan form and the $Q_i$ are quadratic. In the notation of
\cite{kuznetsov10} the elements of $Q_1$ and $Q_2$ are denoted by $a_{ij}$ and
$b_{ij}$, respectively. In the present example it turns out that $a_{20}=0$
and in that case BT.1 and BT.2 are the conditions that $b_{11}\ne 0$ and
$b_{20}\ne 0$. A lengthy calculation shows that $b_{20}=-81W+168k_3>0$ and
$b_{11}=-9W+7k_3>0$, where $W=\frac{256\beta k_4}{(8\beta+1)^3}$. Here we
use the values of the parameters at the bifurcation point. The condition
BT.3 is that the linearization $J_T$ of the mapping
$(x_1,x_2,\delta,k_3)\mapsto(f_1,f_2,{\rm tr}J,\det J)$ at the bifurcation
point is non-singular. This matrix is
\begin{equation}
J_T=\left[
{\begin{array}{cccc}
   12& 48& -k_3 & -\frac5{12}\\
   -3& -12& 0 & 0\\
   -W+\frac53 k_3& -W+k_3& -2k_3& \frac5{12}\\
   9W-17k_3&9W-12k_3& 24 k_3& -2
\end{array}}
\right].
\end{equation}
and $\det J_T=-6k_3^2-\frac{27Wk_3}{4}\ne 0$.

When a generic Bogdanov-Takens bifurcation is present in a dynamical system
then there are always generic Hopf bifurcations nearby. The periodic solutions
which arise in these Hopf bifurcations are hyperbolic and may be stable
(supercritical case) or unstable (subcritical case). We may correspondingly
call the Bogdanov-Takens bifurcation super- or subcritical and it turns out
that these two cases are distinguished by the relative sign of $b_{20}$ and
$b_{11}$. In the present case the signs of these two coefficients are equal
and the bifurcation is subcritical. Hence the periodic solutions are unstable.
In comparison with the phase portrait given in \cite{kuznetsov10}, which 
corresponds to the supercritical case, the direction of the flow is reversed.
For the parameter values for which the bifurcation takes place the cubic
polynomial for $x_1$ has a double root at $x_1=1$ and must therefore have
a factor $(x_1-1)^2$. Carrying out this factorization allows a third root to
be calculated explicitly. The result is
\begin{equation}
p(x_1)=\frac{3}{728}(2457x_1-2924)(x_1-1)^2.
\end{equation}
The additional root is $x_1=\frac{2924}{2457}$ and at the corresponding
steady state $x_2=\frac{995}{4914}$. At that point the trace of the
linearization is negative and the determinant positive. Hence this steady state
is stable.

Some of these results will now be collected in a theorem.

\begin{theorem}\label{th:5}
 There are parameter values for which the system
(\ref{kaimlr1})-(\ref{kaimlr2}) has a generic Bogdanov-Takens bifurcation. In
particular, there are nearby parameter values for which it has an unstable
periodic solution and ones for which it has a homoclinic orbit. In the
case where there is an unstable periodic solution with parameter values
sufficiently close to those at the bifurcation point there are also two
stable steady states and one saddle point.
\end{theorem}

The structural stability of the bifurcation and the fact that the limit
used to obtain this system is normally hyperbolic implies that these features
can be lifted to the system (\ref{kaim1})-(\ref{kaim4}). In more detail,
note first that this system is equivalent by rescaling to the system
(\ref{kaiml1eps})-(\ref{kaiml4eps}). Moreover we can concentrate on a fixed
value of the conserved quantity $X$. Thus it remains to consider a limit from
a three-dimensional system to a two-dimensional one. Restricting to the slow
manifold we get a regular limit of two-dimensional systems. For $\epsilon=0$
the mapping $(x_1,x_2,\delta,k_3)\mapsto(f_1,f_2,{\rm tr}J,\det J)$ has full
rank and a zero at the bifurcation point. It follows from the implicit function
theorem that it has a unique zero near the bifurcation point for $\epsilon$
small. This is a point where BT.0 is satisfied. By continuity BT.1, BT.2 and
BT.3 remain satisfied for $\epsilon$ sufficiently small and the bifurcation
remains subcritical. Thus the features listed in Theorem \ref{th:5} are also seen in the
system on the slow manifold. This implies immediately that there is a
heteroclinic orbit in the full system. The hyperbolic periodic solution in
the slow manifold is also hyperbolic as a solution of the full system. It
is of saddle type with there being both solutions which converge to it for
$t\to +\infty$ and solutions which converge to it for $t\to -\infty$. If it
could be shown that the limiting system admits a stable periodic solution for
some values of the parameters it could be concluded that the full system does
so too. We have not been able to prove the existence of stable periodic
solutions of the limiting system. To see how such a stable solution might
occur, consider the predator-prey model of Bazykin discussed in
\cite{kuznetsov10}. It has two subcritical Bogdanov-Takens bifurcations and a
stable periodic solution in a distant part of the phase space.

It turns out that it is possible to find an extension of the explicit
Bogdanov-Takens point to an explicit two-parameter family of steady states,
including a one-parameter family of points where the eigenvalue condition
for a Hopf bifurcation is satisfied. The parameters are the determinant
$\sigma$ and the trace $\tau$ of $J$ at the steady state. This family is
obtained by fixing the same quantities as in the original case, including
the coordinates $(x_1,x_2)$ of the steady state and computing the parameters
\begin{equation}
k_3=\frac{324+4\sigma+36\tau}{7},\ \
\beta=\frac{132+5\sigma+24\tau}{11232+120\sigma+1248\tau}
\end{equation}
and
\begin{equation}
k_4=\frac{(384+5\sigma+45\tau)(1536+20\sigma+180\tau)}
  {21(1404+15\sigma+156\tau)}.
\end{equation}
The Hopf points are those with $\sigma>0$ and $\tau=0$. These formulae are helpful in
finding parameter values for which the system admits an unstable periodic
solution. A solution of this type is illustrated in Fig.\ \ref{fig:periodicsol} in red for the case
$\sigma=1$, $\tau=-0.02$. In Fig.\ \subref*{fig:periodicsola} we show neighbouring solutions of the unstable periodic solution which move away from it (inward and outward spirals). A larger part of the phase space is shown in Fig.\ \subref*{fig:periodicsolb} where the Bogdanov-Takens point and the stable steady state are shown.

\begin{figure}[htbp]
	\centering
	\subfloat[Unstable periodic solution.]{\includegraphics[width=0.45\textwidth]{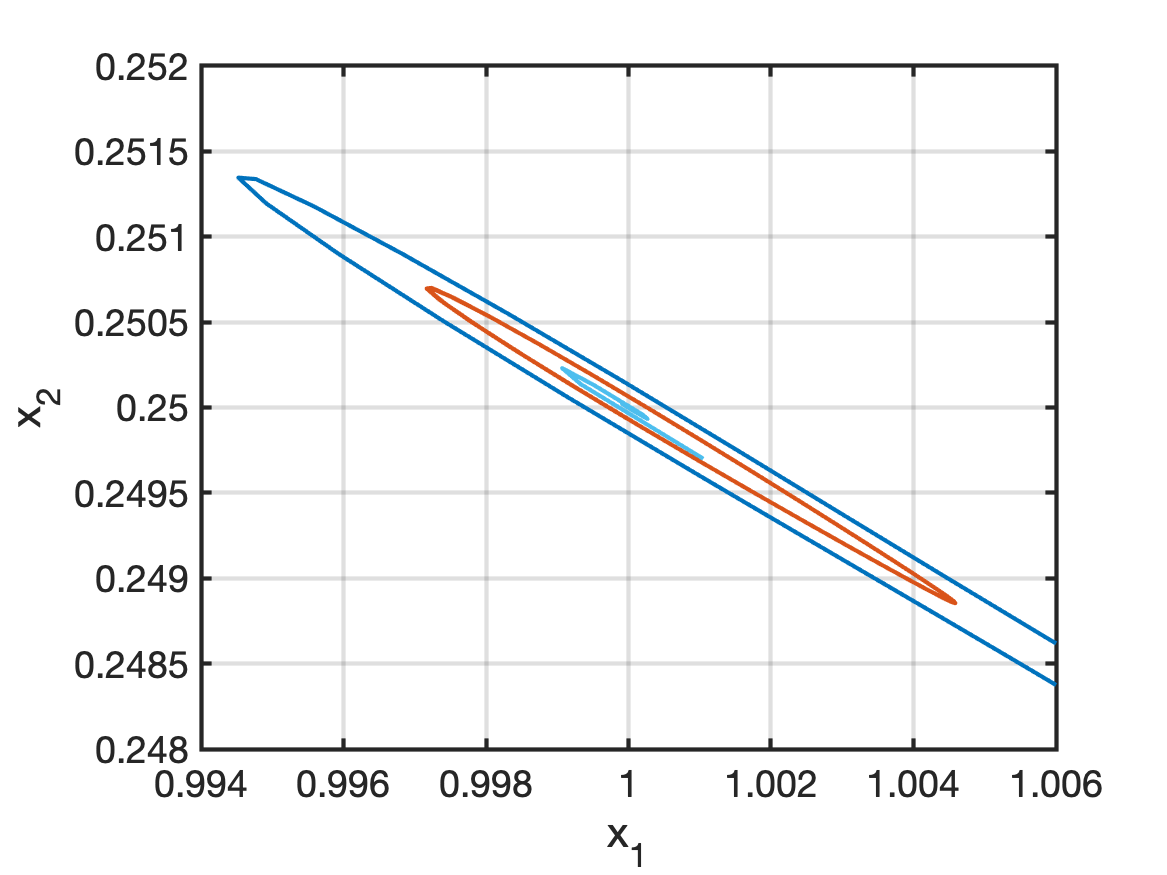}\label{fig:periodicsola}}\hspace*{3em}
	\subfloat[All steady states.]{\includegraphics[width=0.45\textwidth]{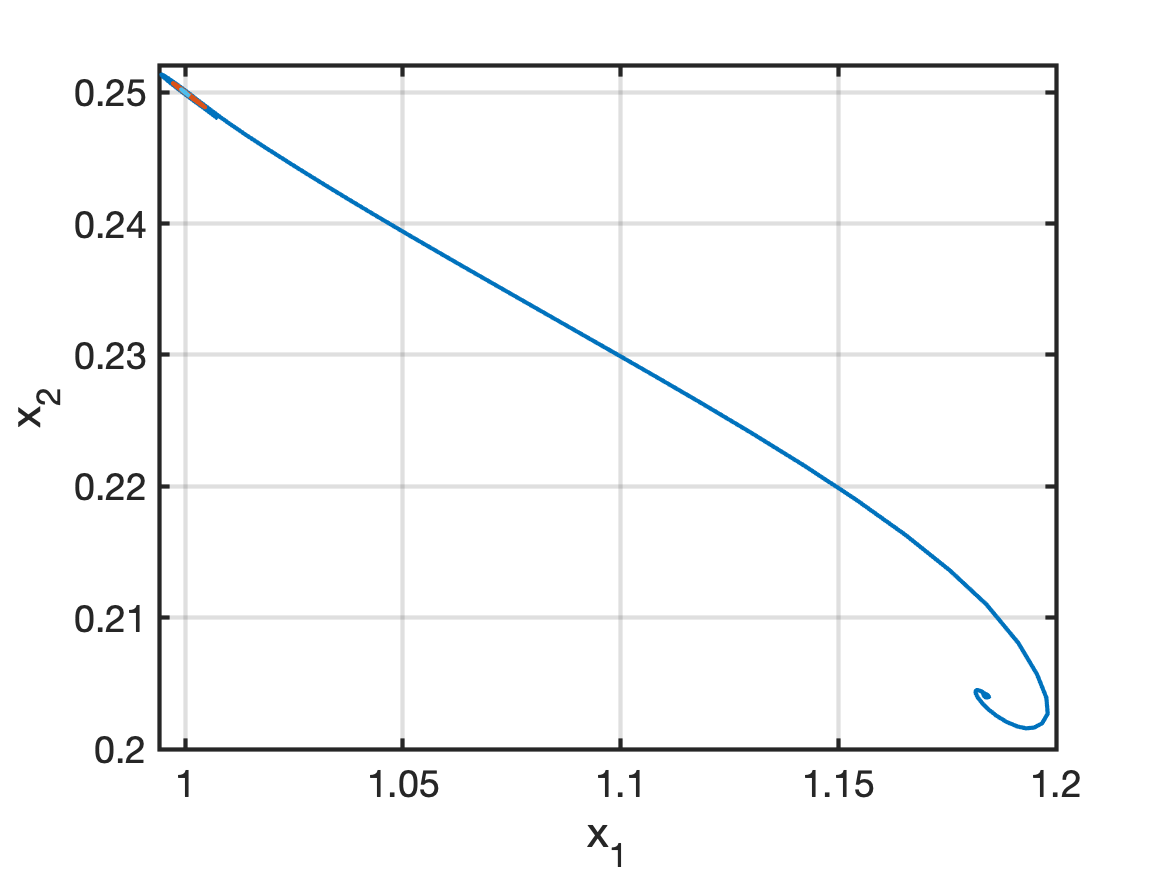}\label{fig:periodicsolb}}
	\caption{Phase diagram for $\sigma=1$, $\tau=-0.02$.}
	\label{fig:periodicsol}
\end{figure}

\section{Comparison with some more elaborate models}\label{sec:6}

In this section the results of this paper will be put into a wider context by
comparing the models analysed here with some more complicated models in the
literature which arise when studying specific biological phenomena. One of
the most exciting recent developments in medicine are immune checkpoint
therapies for cancer \cite{robert20}. Immune checkpoint molecules such as
CTLA4 and
PD-1 can result in the deactivation of T cells under certain circumstances
and this is exploited by cancer cells to evade attacks by the immune system.
Antibodies to the immune checkpoint molecules can prevent this and thus be
used in cancer therapies. This type of therapy has had remarkable success in
curing some cancers. On the other hand, although these therapies could in
principle work for all types of cancer, in practice they only work for some
cancers, notably metastatic melanoma, and even in the most favourable cases
for only a certain percentage of patients.
It is important to obtain a better understanding of the molecular mechanisms
of these therapies, so as to explain in which cases they are effective and,
hopefully, to improve them so as to increase the range of their efficacy.

Up to now the most effective types of immune checkpoint monotherapy are those
involving PD-1. In this context it is important to understand how the
activation of PD-1 leads to suppression of T cell activity. This has been
studied experimentally in \cite{hui17}. The main conclusion of that work is
that the inhibition of T cell activity caused by PD-1 is due less to a
decrease in signalling via the T cell receptor than to a decrease of the
second signal coming from CD28. This leads to a certain mechanistic model of
how the influence of PD-1 is exerted. In an effort to obtain a better
understanding of the mechanisms involved a mathematical model was
introduced in \cite{arulraj18}. Simulations of that model gave results
agreeing well with the results of \cite{hui17} and at the same time suggesting
an additional path by which PD-1 can influence T cell signalling. In the
path highlighted in \cite{hui17} activation of Lck plays an important role.
The suggestion in \cite{arulraj18} is that this change in the activation state
of Lck could have an indirect influence via phosphorylation of molecules
downstream of the T cell receptor and CD28. The model of \cite{arulraj18}
consists of several modules. One of these describes the activation of Lck
and plays a central role.

In the context of their model of Lck regulation the authors of
\cite{arulraj18} cite a model given in \cite{rohrs16}. The latter includes
complexes which are
intermediates in the autophosphorylation reactions. This would correspond
in the case with one phosphorylation site to replacing the reaction
2X$\to$X+Y by the reactions 2X$\to $X${}_2$$\to $X+Y, where X${}_2$ is the
complex formed when two molecules of X bind to each other. It also includes
certain complexes of Lck with Csk which are analogous to XF in the basic
model introduced in section \ref{sec:2}. According to \cite{rohrs16} the inclusion
of these complexes was necessary to obtain a good agreement between the
results of simulations and the experimental data of \cite{hui14}. The
model of \cite{rohrs16} includes no phosphatases and so it is clear that in
that case the evolution must converge to the state where only the unique
maximally phosphorylated state is present. The non-trivial characteristics
of the evolution have to do with the way in which the solution approaches that
state.

One difference of the model of \cite{arulraj18} compared to that of
\cite{kaimachnikov09} is that it includes five forms of Lck rather than four.
The model contains two different forms of doubly phosphorylated Lck which
are supposed to differ by the order in which the two sites were
phosphorylated. The issue of the order of phosphorylation is mentioned in
\cite{hui14} but we are not aware of any justification for including the fifth
form in the model. It is stated in \cite{arulraj18} that the model includes
autophosphorylation of Lck but in the equations the dependence on the
concentrations of the different forms of Lck is everywhere linear and this
does not seem to be consistent.

In \cite{schulze14} the author discusses the model of \cite{kaimachnikov09}
and the alternative where the Michaelis-Menten term in that model is replaced
by a linear one. When that simplification is made bistability is eliminated.
The author presents unpublished data of Acuto and Nika which addresses the
issue of bistability in Lck experimentally. The idea is that if
bistability was present the distribution of the measurements of certain
quantities in a population should be bimodal, i.e. the graph should exhibit
two maxima. In these data most (but not all) of the graphs have a unique maximum
and this is taken as evidence that there is no bistability in the system.
However no detailed justification for this conclusion is given. The
significance of this conclusion is that if bistability were present in the
biological system this would mean that the simplified model would not be
sufficient. In \cite{schulze14} the simplified model is used. The advantage
is that there are less parameters in the simplified model and that their
values can be more strongly constrained by experimental data. 

Another biological phenomenon where Lck plays a central role is that of T
cell activation. It will now be discussed how Lck has been modelled in
the literature on that subject. One of the first and most important
steps in T cell activation is the phosphorylation of the ITAMs (immunoreceptor
tyrosine-based activation motifs) of the T cell receptor complex. The most
important kinase carrying out this process is Lck. In one successful model
of early T cell activation \cite{francois13} Lck is not one of the chemical
species included in the model. In the process of ITAM phosphorylation Lck
is treated as an external kinase whose activity is represented by a reaction
constant. It was proved in \cite{rendall17} that this model can exhibit more
than one steady state.
The model of \cite{francois13} is a radical simplification of a more
extensive one introduced in \cite{altanbonnet05}. In the latter model activated
Lck is one of the chemical species included. It takes part in many reactions
where it binds to a complex X containing the T cell receptor and some other
molecules and then phosphorylates some element of the complex. The kinetics of
these reactions is extended Michaelis-Menten. Other forms of Lck play a role
in mechanisms represented in this model but they do not appear explicitly.
Another model implementing some of the same mechanisms was presented in
\cite{lipniacki08}. It includes four forms of Lck arising from
phosphorylation at Y394 and the serine S59. The serine phosphorylation may
have an important role to play in T cell activation but will not be considered
further here. The tyrosine phosphorylation is supposed to occur by
autophosphorylation in {\it trans} but the Lck molecules responsible for the
catalysis are supposed to belong to a different population to those being
phosphorylated. The former population is treated as external and so no
nonlinearity arises from this process. The model of \cite{lipniacki08}
exhibits bistability.

\section{Conclusions and outlook}\label{sec:7}

In this paper we proved that the model of \cite{doherty15} of an enzyme with a
single site subject to autophosphorylation in {\it trans} can exhibit
bistability. This improves on the simulations in \cite{doherty15} showing
this type of behaviour for specific parameter values by identifying a large
part of parameter space where it occurs. We also show that in the context of
this model multiple steady states can only occur when the phosphorylation
increases the activity of the enzyme. It is shown that in a case where
phosphorylation decreases the activity of the enzyme multiple steady states
can also occur but this requires a more complicated model with an external
kinase which is operating well away from the Michaelis-Menten limit.

We related the models studied in this paper to other models involving Lck
which have been applied in the literature to describe particular biological
phenomena. It is of interest to consider the possible biological meaning of
the results of this paper. Switches arising through bistability are a
well-known phenomenon in biology and the bistability found in the regulation
of Lck might be of importance for immunology as a mechanism by which the
activity of immune cells is switched off or on in certain circumstances. As
discussed in the last section it seems unclear on the basis of experimental
evidence whether bistability due to the properties of Lck occurs in
biologically interesting circumstances. We are not aware that oscillations
in the activation of Lck have been observed experimentally. The
biological significance of those periodic solutions whose existence we
proved is limited by their instability.

This paper is a preliminary exploration of dynamical features of models for
Lck involving autophosphorylation. At this point it is appropriate to think
about what biological issues could be illuminated by continuing these
investigations. A question of great biological and medical interest,
already mentioned in the last section, is that of the mechanism by which
ligation of the receptor PD-1 leads to the suppression of the activity of T
cells. (For a recent review of this topic see \cite{patsoukis20}.) Normally
the activation of a T cell requires both a signal from the T cell receptor
and a second signal from CD28. Both of these receptors get phosphorylated.
(In the case of the T cell receptor it is rather the associated proteins
CD3 and the $\zeta$-chain which are phosphorylated). A question which is
apparently still controversial is whether the main effect of PD-1 activation is
desphosphorylation of the T cell receptor or that of CD28. The conclusion of
\cite{hui17} is that it is CD28 but this has been disputed in \cite{mizuno19},
where it has been suggested that this finding of \cite{hui17} was an artefact
of using a cell-free system and that in reality it is dephosphorylation of
the T cell receptor which is the most important consequence of the activation
of PD-1. This indicates that better understanding of these phenomena is
necessary. The authors of \cite{arulraj18} claim that their model can reproduce
the results of \cite{hui17}. Could that model, or a related one, reproduce the
results of \cite{mizuno19}?

There is a wide consensus that, whatever the targets of dephosphorylation
resulting from the activation of PD-1, the phosphatase which carries it out
is SHP-2. There is one caveat here since it was observed in \cite{rota18}
that PD-1 can have an inhibitory effect on T cells in the absence of
SHP-2. This issue deserves further investigation. Another interesting
question is that of the way in which Lck interacts with SHP-2 and PD-1.
When PD-1 is fully activated it is phosphorylated at two sites. These
provide binding sites for SHP-2. The phosphorylation of PD-1 is catalysed
primarily by Lck \cite{hui17}. SHP-2 can dephosphorylate PD-1 and thus
promote its own unbinding. This effect is opposed by Lck. Here there is
an incoherent feed-forward loop \cite{alon06}. On the one hand Lck causes
phosphorylation of PD-1 by a direct route and on the other hand it causes its
dephosphorylation by an indirect route. These interactions are described by
one of the modules in the model of \cite{arulraj18}. They are sufficiently
complex that it would be desirable to carry out a deeper mathematical
analysis of their dynamics.

\section*{Acknowledgments}  LMK acknowledges support from the European Union Horizon 2020 research and innovation programmes under the Marie Sk\l odowska-Curie grant 
agreement No.\ 777826 (NoMADS), the Cantab Capital Institute for the Mathematics of Information and Magdalene College, Cambridge (Nevile Research Fellowship).

\end{document}